\documentclass{article}

     \usepackage[preprint, nonatbib]{neurips_2020}

\usepackage[utf8]{inputenc} 
\usepackage[T1]{fontenc}    
\usepackage{hyperref}       
\usepackage{url}            
\usepackage{booktabs}       
\usepackage{amsfonts}       
\usepackage{nicefrac}       
\usepackage{microtype}      
\usepackage{amsmath}
\usepackage{graphicx}
\usepackage{amsthm}
\usepackage{subfig}
\usepackage{textcomp}
\usepackage{dsfont} 
\usepackage{wrapfig}
\usepackage[numbers]{natbib}

\usepackage[ruled, vlined]{algorithm2e}
\usepackage{etoolbox}

\makeatletter
\patchcmd{\@algocf@start}
  {-1.5em}
  {0pt}
  {}{}
\makeatother

\newcommand{\norm}[1]{\Vert#1\Vert}
\newcommand{\universe}{\mathbb{R}_{\geq 0}^{d}}

\DeclareMathOperator{\bindist}{Bin}
\DeclareMathOperator{\poidist}{Poi}
\DeclareMathOperator{\unifdist}{Unif}
\newcommand{\ip}[2]{\langle{#1},{#2}\rangle}

\newtheorem{theorem}{Theorem}
\newtheorem{lemma}[theorem]{Lemma}

\SetKwFor{cFor}{for (}{) $\lbrace$}{$\rbrace$}

\title{DartMinHash: Fast Sketching for Weighted Sets}

\author{%
  Tobias Christiani \\
  Department of Computer Science \\
  Norwegian University of Science and Technology \\
  Trondheim, Norway \\
  \texttt{tobias.christiani@ntnu.no}
}

\begin{document}

\maketitle

\begin{abstract}
  Weighted minwise hashing is a standard dimensionality reduction technique with applications to similarity search and large-scale kernel machines.
  We introduce a simple algorithm that takes a weighted set $x \in \universe$ and computes $k$ independent minhashes in expected time $O(k \log k + \norm{x}_{0}\log(\norm{x}_1 + 1/\norm{x}_1))$, 
  improving upon the state-of-the-art BagMinHash algorithm (KDD \textquotesingle 18) and representing the fastest weighted minhash algorithm for sparse data. 
  Our experiments show running times that scale better with $k$ and $\norm{x}_0$ compared to ICWS (ICDM \textquotesingle 10) and BagMinhash, obtaining $10$x speedups in common use cases.
  Our approach also gives rise to a technique for computing fully independent locality-sensitive hash values for $(L, K)$-parameterized approximate near neighbor search under weighted Jaccard similarity in optimal expected time $O(LK + \norm{x}_0)$, improving on prior work even in the case of unweighted sets.
\end{abstract}

\section{Introduction}
Consider a weighted set $x \in [0, M]^d$ with elements $i \in \{1, \dots, d\}$ and weights $0 \leq x_i \leq M$.
A simple scheme for computing the weighted minhash $h(x)$ of $x$ is to throw darts uniformly at random into $[0, M]^d$ and setting $h(x)$ to the rank (index) of the first dart that hits $x$~\cite{charikar2002,shrivastava2016simple}.
Two weighted sets hash to the same value if and only if the first dart that hits either set lands in the intersection of the two sets.
The probability of collision (taken over the random sequence of darts) is therefore exactly equal to the weighted Jaccard similarity $J(x, y)$ between the two sets.
\begin{equation} \label{eq:minhash}
  \Pr[h(x) = h(y)] = J(x, y) = \frac{\sum_i \min(x_i, y_i)}{\sum_i \max(x_i, y_i)}
\end{equation}
Weighted minwise hashing is also known as consistent weighted sampling~\cite{manasse2010consistent} and a consistent weighted sample from a set $x$ simply corresponds to the location (element and weight) of the first dart hitting $x$.
For notation we use the standard $\ell_p$-norms with $\norm{x}_0$ denoting the number of elements of $x$ with non-zero weight and $\norm{x}_1$ denoting the sum of the weights of $x$.

The straightforward rejection sampling approach of throwing darts until one hits requires throwing an expected $dM/\norm{x}_1$ darts to compute a minhash value, making it slow for sparse data~\cite{shrivastava2016simple}.
The main technical contribution of this paper is an improved hashing-based algorithm for efficiently recovering the first $t$ darts that hit $x$ with expected running time $O(t + \norm{x}_0 \log(\norm{x}_1 + 1/\norm{x}_1))$ that is independent of the sparsity of $x$.
For sets of constant weight $\norm{x}_1 = \Theta(1)$ this matches the optimal running time $O(t + \norm{x}_0)$ of reading the $\norm{x}_0$ elements of the input and returning $t$ darts.

Having fast access to the first $t$ darts hitting a set is a sketching primitive that can be used to compute independent minhashes~\cite{dahlgaard2017fast,ertl2018bagminhash}, 
weighted bottom-$k$ and one-permutation minhashes~\cite{cohen2007summarizing, li2012one},
as well as faster locality-sensitive hash functions for nearest neighbor search~\cite{indyk1998}.
This paper focuses on the metric-driven approach to sketching, where the metric (or similarity) is given and we are interested in developing efficient sketching algorithms with theoretical guarantees. 
\subsection{Applications of weighted minwise hashing}
%
\paragraph{Similarity estimation.}
Perhaps the most direct application is the use of minhash values as sketches for estimating the Jaccard similarity between weighted sets~\cite{broder1997syntactic}. 
Using $k$ independent minhashes $h_1(x), \dots, h_k(x)$ to sketch a weighted set $x$, and applying the same hash functions $h_1(y), \dots, h_k(y)$ to form a sketch of $y$, 
we estimate the weighted Jaccard similarity between $x$ and $y$ from their sketches using the estimator $\hat{J}(x, y) = (1/k)\sum_j \mathds{1}[h_j(x) = h_j(y)]$ where $k\hat{J}(x, y) \sim \bindist(k, J(x, y))$.
Sketches can be further compressed by mapping each minhash randomly to a $b$-bit fingerprint~\cite{charikar2002, li2011theory}.
Using $1$-bit minhash we can pack $64$ independent minhash values into a single machine word with the collision probability of each bit being $(1 + J(x, y))/2$.
Typical sketch lengths lie in the ranges between $k = 64$ and $k = 1024$~\cite{li2011theory, mitzenmacher2014efficient, chakrabarti2015bayesian}. 
\paragraph{Similarity search.}
Weighted minwise hashing is an example of a family of Locality-Sensitive Hash (LSH) functions which makes it applicable in standard solutions to a large number of similarity search problems, including exact and approximate nearest neighbor search~\cite{indyk1998,har-peled2012, indyk2004, aumuller2019puffinn}.
We can use locality-sensitive hashing to preprocess a dataset in order to support fast nearest neighbor queries by placing data points into buckets according to their locality-sensitive hash values.
During queries we only consider the subset of data points that hash to the same buckets as the query, speeding up the search compared to a linear scan.
The query/preprocessing operation searches/stores each point in $L$ buckets according to the concatenation of $K$ independent locality-sensitive hash values.
Typically $L$ is between $20-500$ and $K$ is between $10-30$, making the computation of $LK$ independent weighted minwise hash values a bottleneck, 
and motivating a significant research effort to speed up various LSH schemes as well as reducing the number of independent LSH computations required to perform nearest neighbor queries~\cite{shrivastava2014densifying, andoni2015practical, dahlgaard2017fast, christiani2019fast, aumuller2019puffinn}.

\paragraph{Large-scale kernel machines.}
Weighted minwise hashing also been successfully applied to speed up large-scale kernel machines~\cite{li2011theory, li2012one, li2017linearized}.
Similarly to the application of random Fourier features to linearize shift-invariant kernels such as the standard RBF kernel~\cite{rahimi2007}, 
weighted minwise hashing can be used to linearize the kernel $J(x, y)$ through a unary mapping $u$ of $b$-bit minhash values such that $\ip{u(h(x))}{u(h(y))} \approx J(x, y)$.
This drastically speeds up learning on large datasets as it avoids expensive Gram matrix computations by allowing us to use a linear SVM directly on the transformed data.
In experiments on a large number of classification tasks the $J(x, y)$ kernel has shown competitive and in many cases superior accuracy compared with the RBF kernel~\cite{li2017linearized}.
Furthermore, this performance carries over when we linearize the kernel, using as few as $k = 256$ minhashes and comparing favorably with random Fourier features which in turn are more expensive to compute. 
Once a linear SVM has been trained, the classification of a new data point is a simple as computing an inner product, thus making the linearization/randomized embedding by minhashing the bottleneck, which further motivates the need for fast weighted minwise hashing algorithms. 
\subsection{Related work}
\paragraph{Consistent weighted sampling.}
Since the original minhash algorithm for discrete sets~\cite{broder1997syntactic}, there has been a line of work attempting to achieve the same performance for weighted sets 
culminating in the (improved) consistent weighted sampling algorithm with running time $O(\norm{x}_0)$ for computing a single minhash value~\cite{gollapudi2006exploiting, manasse2010consistent, ioffe2010improved}.
A key insigt behind these algorithms is for a given weighted element $x_i$ we only need to consider the location $v \leq x_i$ of the dart of minimal rank that has hit $x_i$, the so-called ``active index''.
We can then exploit the fact that all weights between $v$ and the next active index have the same minhash value, and that the distance between active indices follows a known distribution.  

\paragraph{Fast techniques for unweighted sets.}
Since many applications require computing several hundred minhash values, there has also been significant efforts to speed up both discrete and weighted minhash algorithms.
In the case of discrete sets (all weights are either zero or one) the technique of one-permutation hashing~\cite{li2012one} is able to achieve running time $O(k + \norm{x}_0)$, but the $k$ hash values are not independent and the scheme cannot easily deal with sets of different sizes~\cite{shrivastava2014densifying, li2019rerandomized}. 
The authors of~\cite{dahlgaard2017fast} introduced fast similarity sketching as an approach to computing $k$ minhashes satisfying strong Chernoff-style concentration bounds in expected time $O(k \log k + \norm{x}_0)$ in the discrete case, overcoming some of the drawbacks of one-permutation hashing.
The problem with one-permutation hashing stems from the fact that the standard approach to minhashing in the discrete case is to ``throw darts'' or sample elements from the universe \emph{without replacement} in order to create a permutation.
Thus when processing a discrete set with $\norm{x}_0 < k$ elements, we get fewer than $k$ darts to produce minhash values from, and these values have dependencies since we are sampling without replacement. 
The fast similarity sketching approach of Dahlgaard et al.~\cite{dahlgaard2017fast} overcomes this by switching to sampling with replacement as needed in order to produce $k$ minhash values, although they do not manage to show full independence.
Algorithms for the discrete case can be applied to weighted sets by discretizing the weights at the cost of a discretization error and increased running times, losing general applicability~\cite{gollapudi2006exploiting, haeupler2014consistent}.

\paragraph{Fast weighted minwise hashing.}
For the weighted case the simple rejection sampling approach described earlier has been shown to work well in practice~\cite{shrivastava2016simple}.
The drawbacks of using rejection sampling directly is that it requires relatively tight a priori bounds on the weights of each element and that its performance degrades with the sparsity of the data.
In \cite{li2019rerandomized} the authors consider a combination of one-permutation hashing and Consistent Weighted Sampling (CWS) by first hashing the elements of a weighted set $x$ into a number of bins and then applying CWS on each bin.
This heuristic approach is shown to work well in practice, but it requires that the weight of $x$ is close to uniformly distributed across elements and it lacks general guarantees.

The BagMinHash algorithm~\cite{ertl2018bagminhash} combines ideas from fast similarity sketching with the ``active indices'' idea from the consistent weighted sampling approach to construct a fast algorithm for fully independent weighted minwise hashing.
The BagMinHash algorithm produces $k$ minhashes from a weighted set by simultaneously for each weighted element $x_i$ performing a top-down binary search for the relevant active indices of $k$ independent sequences of darts.
By generating the active indices in increasing order of their rank and keeping track of the max rank of the hitting darts found in each of the $k$ sequences the search can be stopped early, thus avoiding generating $k$ active indices for each weighted element.
The BagMinHash paper lacks an explicit running time bound, but a bound of $O(f(k, \omega) + \norm{x}_0 \omega \log \omega)$ is given where $f$ is some function and $\omega$ denotes the bit-length of the (floating-point) representation of weights $x_i$. 
Experimentally, the BagMinHash algorithm is shown to be faster than ICWS for $\norm{x}_0 \geq 100$ and $k \geq 256$.
\subsection{Contributions}
%
\begin{wrapfigure}[19]{R}{0.4\textwidth}
  \centering
  \includegraphics[width=0.4\textwidth]{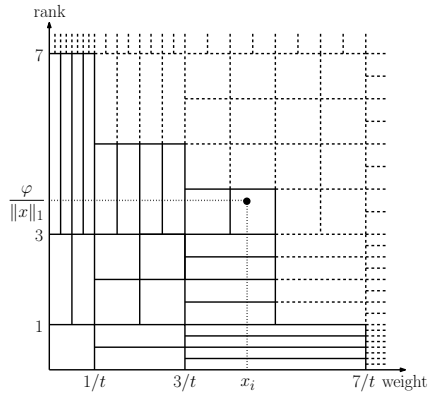}
  \caption{Weight-rank subdivision with areas relevant for $(\varphi/\norm{x}_1, x_i)$ marked in solid} \label{fig:weight-rank}
\end{wrapfigure}
\paragraph{DartHash.}
We develop an efficient algorithm for recovering the first $t$ darts hitting a weighted by essentially turning the rejection sampling approach on its head.
Instead of throwing darts uniformly at random into $[0, M]^d$ and keeping the first $t$ that hit our set, 
we could preprocess the darts into a data structure supporting queries of the form: ``Among the first $r$ darts thrown, return those hitting $x_i$''.
One of the challenges with this approach is that it would require storing a large number of darts in order to support queries on sparse sets where we need to throw $t d M / \norm{x}_1$ darts in order to have $t$ darts hitting $x$ in expectation.
Another challenge is how to subdivide the darts into ranges of ranks and weights in order to support efficient queries while ensuring that the results for all sets are consistent with the same underlying sequence of darts.
In Section \ref{sec:darthash} we introduce the DartHash algorithm for solving this problem with expected running time $O(t + \norm{x}_0 \log(\norm{x}_1 + 1/\norm{x}_1))$.
The key observation we use is that asymptotically as we increase the size of the universe ($d$ and $M$) the number of darts (from an infinite sequence) with ranks in a certain range hitting an element in a certain range of weights is Poisson distributed and independent of the number of darts hitting other such areas.
This makes it possible to use random hashing to obtain the darts hitting relevant areas of the weight-rank space, and by using a particular pattern to subdivide the weight-rank space in order to support fast queries we obtain our result. 

\paragraph{Weighed minwise hashing.}
Once we have an efficient method of retrieving the first hitting darts we can easily create weighted minhash values.
As shown earlier, the rank of the first hitting dart is a valid minhash value, so if we had the first dart in each of $k$ independent sequences of darts we would have $k$ minhash values.
Since the darts returned by the DartHash algorithm are independent, creating $k$ independent sequences is as simple as using a random hash function to assign each dart to one of $k$ sequences.
Using the DartHash algorithm to create $k$ weighted minhashes we need to recover the first $t = O(k \log k)$ darts in expectation to ensure that there is at least one dart assigned to each of the $k$ sequences.
The resulting ``DartMinHash'' algorithm produces $k$ independent weighted minhash values in expected time $O(k \log k + \norm{x}_0 \log(\norm{x}_1 + 1/\norm{x}_1))$.
\begin{wrapfigure}[26]{R}{0.5\textwidth}
  \begin{algorithm}[H]
    \small
    \SetKwInOut{Params}{Parameters}
    \SetKwArray{Darts}{darts}
    \SetKw{And}{and}
    \DontPrintSemicolon
    $D \leftarrow \emptyset $ \; 
    \For{$i \in \{ j \mid x_j > 0 \}$} { 
      \For{$\nu \leftarrow 0$ \KwTo $\lfloor \log_2(1+tx_i) \rfloor$} {
        \For{$\rho \leftarrow 0$ \KwTo $\lfloor\log_2(1+\varphi/\norm{x}_1)\rfloor$} {
          $W \leftarrow (2^\nu - 1)/t$, $R \leftarrow 2^\rho - 1$ \;
          $\delta_{\nu} \leftarrow 2^{\nu}/t2^{\rho}$, $\delta_{\rho} \leftarrow 2^{\rho}/2^{\nu}$ \;
          \For{$w \leftarrow 0$ \KwTo $2^{\rho}-1$} {
            \lIf{$x_i < W + w  \delta_{\nu}$} {break}
            \For{$r \leftarrow 0$ \KwTo $2^{\nu}-1$} {
              \lIf{$\varphi/\norm{x}_1 < R + r  \delta_{\rho}$} {break}
              $j \leftarrow 0$, $X \sim \poidist(1)$\;
              \While{$j < X$}{
                $V, U \sim \unifdist[0, 1)$ \;
                $weight \leftarrow W + (w + V)  \delta_{\nu}$\;
                $rank \leftarrow R + (r + U)  \delta_{\rho}$\;
                $index \leftarrow (i, \nu, \rho, w, r, j)$\;
                \If{$weight \leq x_i$ {\normalfont and} $rank \leq \varphi/\norm{x}_1$} {
                  $D \leftarrow D \cup (index, rank)$ 
                }
                $j \leftarrow j + 1$\;
              }
            }
          }
        }
      }
    }
    \Return{$D$}
    \caption{$\textsc{DartHash}_{t}(x, \varphi)$} \label{alg:darthash}
  \end{algorithm}
\end{wrapfigure}
\paragraph{Fast similarity search.}
DartHash can also be used to create the $LK$ locality-sensitive hash values for the LSH solution to approximate near neighbor search in optimal expected time $O(LK + \norm{x}_0)$, both improving prior work for discrete sets~\cite{shrivastava2014densifying,dahlgaard2017fast} and extending it to the weighted case.
For approximate near neighbor search, the important property of the $K$ locality-sensitive hash values used in each of the $L$ tables is that the probability of $x$ and $y$ colliding equals $J(x, y)^K$ and is independent between tables.
Consider the first $K$ darts hitting weighted sets $x$ and $y$. The probability that they have the same first hitting dart is exactly equal to $J(x,y)$.
Contingent on the first dart colliding, the probability of the two sets sharing the same second dart is exactly $J(x, y)$ since we have essentially reset the process after the first collision, and so on.
We can therefore use the first $K$ darts in each of $L$ independent sequences as our hash values and because $K = \Theta(\log L)$ the complexity is the same as using DartHash to obtain the first $LK$ darts.
By further dividing our data points into weight classes and normalizing to $\norm{x}_1 = \Theta(1)$ we avoid the additional overhead of dealing with large or small weights.
Even if we only consider the special case of discrete sets, this is the first result that is able to match the guarantees of using fully independent minhashes in expected time $O(LK + \norm{x}_0)$.
Previously, both one-permutation hashing~\cite{li2012one,shrivastava2014densifying} and fast similarity sketching~\cite{dahlgaard2017fast} has been able to achieve the same complexity for discrete sets, 
but lacking important guarantees from using fully independent hash values that allow us to upper bound the probability of not finding $y$ when querying for $x$ by $(1 - J(x, y)^K)^L$,
instead having to resort to weaker variance-based bounds~\cite{dahlgaard2017fast,christiani2019fast}.
\section{DartHash} \label{sec:darthash}
The DartHash algorithm uses a structured subdivision of the range of weights of each element and the infinite sequence of darts in order to efficiently simulate the rejection-sampling approach to finding the first darts hitting a set.
The algorithm is initialized with a parameter $t$ that is used to control how the weights are subdivided/hashed.
In addition, a call to $\textsc{DartHash}_{t}(x, \varphi)$ takes an argument $\varphi > 0$ that controls the upper limit on the rank of darts to return, 
so that the algorithm returns $\varphi t$ darts in expectation.
The remainder of this section will prove our main Lemma: 
\begin{lemma} \label{lem:darthash}
  $\textsc{DartHash}_t(x, \varphi)$ returns the first $D \sim Poi(\varphi t)$ darts hitting $x$ in expected time $O(\varphi t +  \norm{x}_0 \log(\norm{x}_1/\varphi + \varphi/\norm{x}_1))$.
\end{lemma}
\paragraph{Poisson darts.}
Consider an infinite sequence of darts thrown uniformly at random into $[0, M]^d$ in the sense that the location of each dart is determined by sampling a uniform element from $\{1, \dots, d\}$ and a uniform weight from $[0, M]$.
By a standard argument the number of darts among the first $dM/\norm{x}_1$ that hit a weighted set $x$ follows a Poisson distribution in the limit as $dM$ goes to infinity~\cite{mitzenmacher2017probability}.
We are interested in recovering the first $\varphi t$ darts hitting $x$, so we need to consider darts with rank at most $\varphi tdM/\norm{x}_1$. 
To simplify the exposition we will henceforth refer use the term rank to refer to the normalized rank where we have divided through by $tdM$, 
i.e. we have an expected $\varphi t$ darts of rank at most $\varphi / \norm{x}_1$ hitting $x$.

Poisson distributed random variables can be split and combined~\cite{mitzenmacher2017probability}.
This allows us to subdivide the number of darts hitting an element $x_i$ into areas spanning different ranges of weights and ranks.
Figure~\ref{fig:weight-rank} shows the doubling-pattern employed in the DartHash algorithm which has been chosen specifically to minimize the running time for sets of unit weight. 
The weight-rank space for each element is subdivided into a number of \emph{regions} each containing a number of \emph{areas}. 
Region $(\nu, \rho)$ covers the portion of space given by $[(2^\nu - 1)/t, (2^{\nu+1} - 1)/t) \times [2^\rho - 1, 2^{\rho + 1} - 1)$.
Each region $(\nu, \rho)$ is split evenly into $2^\rho$ vertical subdivisions and $2^\nu$ horizontal subdivisions, resulting in a total of $2^{\nu + \rho}$ areas within region $(\nu, \rho)$.

The DartHash algorithm (see Algorithm \ref{alg:darthash} for pseudocode) works by going through each nonzero element $x_i$ of $x$, and iterating through the relevant regions and areas of $x_i$ as shown in Figure~\ref{fig:weight-rank} to find all darts of rank at most $\varphi/\norm{x}_1$ hitting $x_i$.
To avoid explicitly storing all the darts we use random hash functions to simulate draws from the Poisson distribution, determining the number of darts in an area, as well as their exact weight and rank within the area.
For the sake of readability we have omitted the subscripts from the random variables in Algorithm~\ref{alg:darthash} when in fact the random variables are tied to the particular element and area, e.g. $X_{i, \nu, \rho, w, r}$, and consistent between different weighted sets. 
The expected running time of finding the darts hitting $x_i$ can be upper bounded by the number of areas inspected by the algorithm since each area only contains one dart in expectation.
\begin{lemma} \label{lem:element_bound}
  A single iteration of the outer loop of Algorithm \ref{alg:darthash} 
  has expected running time $O(t x_i \varphi / \norm{x}_1 + \log(1 + tx_i) + \log(1 + \varphi/\norm{x}_1))$.
\end{lemma}
\begin{proof}
Let $A = \lfloor \log_2(1+tx_i) \rfloor$ and $B = \lfloor\log_2(1+\varphi/\norm{x}_1)\rfloor$.
The algorithm investigates areas in regions $(\nu, \rho) \in \{0,\dots, A\} \times \{0, \dots, B\}$. 
We proceed by considering the four cases: (i) $A, B > 0$, (ii) $A > 0, B = 0$, (iii) $A = 0, B > 0$, and (iv) $A, B = 0$.
In case (i) $A, B > 0$ we upper bound the number of visited areas by the total number of areas in the relevant regions 
$\sum_{\nu = 0}^{A} \sum_{\rho = 0}^{B} 2^{\nu + \rho} = O(2^{A + B})$.
In case (ii) $A > 0, B = 0$ we get the bound $\sum_{\nu = 0}^{A} \lceil 2^{\nu} (\varphi/\norm{x}_1) \rceil = O(2^{A} (\varphi/\norm{x}_1) + A)$.
Case (iii) $A = 0, B > 0$ results in a similar bound of $\sum_{\rho = 0}^{B} \lceil 2^{\rho} tx_i \rceil = O(2^{B} tx_i  + B)$.
The case (iv) $A, B = 0$ is trivially bounded by $1$ since the region $(0, 0)$ only contains a single area.
The number of areas in each case is thus bounded by the expression $O(\min(1, \varphi/\norm{x}_1) \cdot \min(1, tx_i) 2^{A+B} + A + B)$.
Lemma \ref{lem:element_bound} follows from using that $\min(1, \varphi/\norm{x}_1) 2^{A} \leq \min(1, \varphi/\norm{x}_1)(1 + \varphi/\norm{x}_1) \leq 2\varphi/\norm{x}_1$ 
and similarly that $\min(1, tx_i) 2^{B} \leq 2tx_i$.
\end{proof}
To arrive at Lemma \ref{lem:darthash} we first sum the expression in Lemma \ref{lem:element_bound} for the elements of $x$ to obtain the bound $O(\varphi t + \sum_i \log(1 + tx_i) + \norm{x}_0 \log(1 + \varphi/\norm{x}_1))$.
Next, we note that $\sum_i \log(1 + tx_i) \leq \norm{x}_0 \log(1 + t\norm{x}_1/\norm{x}_0)$ since the expression on the left is maximized when we spread the weight $\norm{x}_1$ of $x$ evenly across its $\norm{x}_0$ nonzero elements.
Finally, we can show that $\norm{x}_0 \log(1 + t\norm{x}_1/\norm{x}_0) = O(\varphi t + \norm{x}_0 \log(1 + \norm{x}_1/\varphi))$ and Lemma \ref{lem:darthash} follows.
This last bound is trivially true for $\norm{x}_0 \geq \varphi t$. 
Otherwise, for $\norm{x}_0 = \gamma \varphi t$ where $\gamma \in (0,1)$ and in the case where $\norm{x}_1 / \gamma \varphi \leq 1$ we have that $\gamma \varphi t \log(1 + \norm{x}_1 / \gamma \varphi) \leq \varphi t$.
For $\norm{x}_1 / \gamma \varphi > 1$ we have $\gamma \varphi t \log(1 + \norm{x}_1 / \gamma \varphi) \leq \gamma \varphi t (\log(\norm{x}_1 / \varphi) + \log(1/\gamma)) + \varphi t$ and we can use that $\gamma \log(1/\gamma) \leq 1/e$.
\section{Fast sketching}
In this section we will use the DartHash algorithm to produce independent weighted minhashes, bottom-$k$ weighted minhashes, and fast locality-sensitive hash values for nearest neighbor search.
We will make central use of the trick of randomly hashing the darts from the $\textsc{DartHash}_t(x, \varphi)$ algorithm to $k$ sequences, 
forming $k$ independent Poisson processes with rate $\varphi t/k$~\cite{mitzenmacher2017probability}.
Having $k$ independent sequences we can use the dart with smallest rank in each sequence to produce a minhash value. 
This approach of hashing darts to arrive at minhash values was used previously in~\cite{dahlgaard2017fast, ertl2018bagminhash}.
\begin{theorem} \label{thm:minhash}
  Let $k$ be a positive integer and assume access to a constant-time fully random hash function, 
  then there exists an algorithm that takes as input a weighted set $x \in \universe$ and outputs $k$ independent minhash values
  using expected time $O(k \log k + \norm{x}_0 \log(\norm{x}_1 + 1/\norm{x}_1))$.
\end{theorem}
\begin{proof}
  To create $k$ minhashes we will proceed by evaluating $\textsc{DartHash}_t(x, \varphi)$ in steps $\varphi = 1,2,\dots$ and randomly hashing the resulting darts to $k$ buckets, stopping once all $k$ buckets have at least one dart.
  The $j$th minhash value is the dart with the smallest rank in bucket $j$.
  By fixing $t = \lceil k \ln k \rceil + k$ the number of darts in each bucket at step $\varphi$ is Poisson distributed with rate $\lambda \geq \varphi(1 +\ln k)$.
  We can union bound the probability that all $k$ buckets are empty after step $\varphi$ by $k e^{-\varphi(1 + \ln k)} \leq e^{-\varphi}$.
  The expected running time can be written as a sum over the probability of each step times its cost: $\sum_{\varphi = 1}^{\infty} e^{-(\varphi - 1)}O(\varphi k \log k +  \norm{x}_0 \log(\norm{x}_1/\varphi + \varphi/\norm{x}_1))$.
  For $\varphi \geq 1$ we can upper bound the cost of step $\varphi$ by $\varphi \cdot O(k \log k + \log(\norm{x}_1 + 1/\norm{x}_1))$. 
  Finally, we can use that $\sum_{\varphi = 1}^{\infty} \varphi e^{-(\varphi - 1)}  \leq e \sum_{\varphi = 1}^{\infty} \varphi 2^{-\varphi} = 2e$ and Theorem~\ref{thm:minhash} follows.  
\end{proof}
\paragraph{Bottom-$k$ minhash and fast approximate near neighbor search.}
The bottom-$k$ sketch consists of the first $k$ darts hitting a set.
Compared to $k$ independent minhashes it provides a more accurate estimate of the Jaccard similarity between sets~\cite{cohen2007summarizing}.
The \textsc{DartHash} algorithm can compute bottom-$k$ weighted minhashes in expected time $O(k + \norm{x}_0 \log(\norm{x}_1 + 1/\norm{x}_1))$ by setting $t = k$ and increasing $\varphi$ until at least $k$ darts are found.
These first $k$ darts can be returned in sorted order in expected time $O(k)$ using a bucket sort on their rank.

To speed up the standard locality-sensitive hashing solution to approximate nearest neighbor search under weighted Jaccard similarity we first split the data $P$ into weight classes e.g. $\norm{x}_1 \in [1, 2), [2, 4), [4, 8),..$.
Within each weight class we normalize data to have weight $\norm{x}_1 \in [1, 2)$. 
Given a query point $q$ we are interested in finding points $x \in P$ with $J(q, y) \geq j_1$ for some constant e.g. $j_1 = 0.5$, allowing us to restrict our search to a constant number of different weight classes where our query point will be normalized to $\norm{q}_1 = \Theta(1)$. 
In each weight class we can use DartHash with $t = LK$ and randomly hash the darts into $L$ buckets to create $L$ independent sorted bottom-$K$ minhash sketches with collision probability $J(q, x)^K$ that can be used as locality-sensitive hash values .
Since the standard LSH data structure has $K = \Theta(\log L)$ we can use Chernoff bounds to show that we only need to recover $O(LK)$ darts in expectation to fill the $L$ bottom-$K$ sketches.
\section{Experiments}
Our experiments will focus on investigating the correctness and performance of the DartMinHash algorithm compared to BagMinHash~\cite{ertl2018bagminhash} and ICWS~\cite{ioffe2010improved}.
To corroborate the claim that DartMinHash returns independent weighted minhash values we will compare similarity estimates for different sketch lenghts to their theoretical distribution under full independence.
To study the performance of different algorithms we measure running times on synthetic data when varying $k$, $\norm{x}_0, \norm{x}_1$. 

\paragraph{Implementation.}
For our experiments we have implemented the darthash algorithm and its related sketches in C\texttt{++}.\footnote{The source code is available on GitHub: \url{https://github.com/tobc/dartminhash}.}
The implementation closely follows the pseudocode in Algorithm~\ref{alg:darthash} with a few additional optimizations such as using precomputed tables for the powers of two and the Poisson CDF to avoid expensive powering and division operations.
To simulate random draws we use tabulation hashing~\cite{zobrist1970new} seeded using the Mersenne Twister PRNG~\cite{matsumoto1998mersenne}.
Tabulation hashing is fast to compute and has been shown to have strong pseudorandomness properties in a variety of applications~\cite{patrascu2012power, dahlgaard2017practical}.
\begin{wrapfigure}[21]{R}{0.45\columnwidth}
  \centering
  \includegraphics[width=0.45\columnwidth]{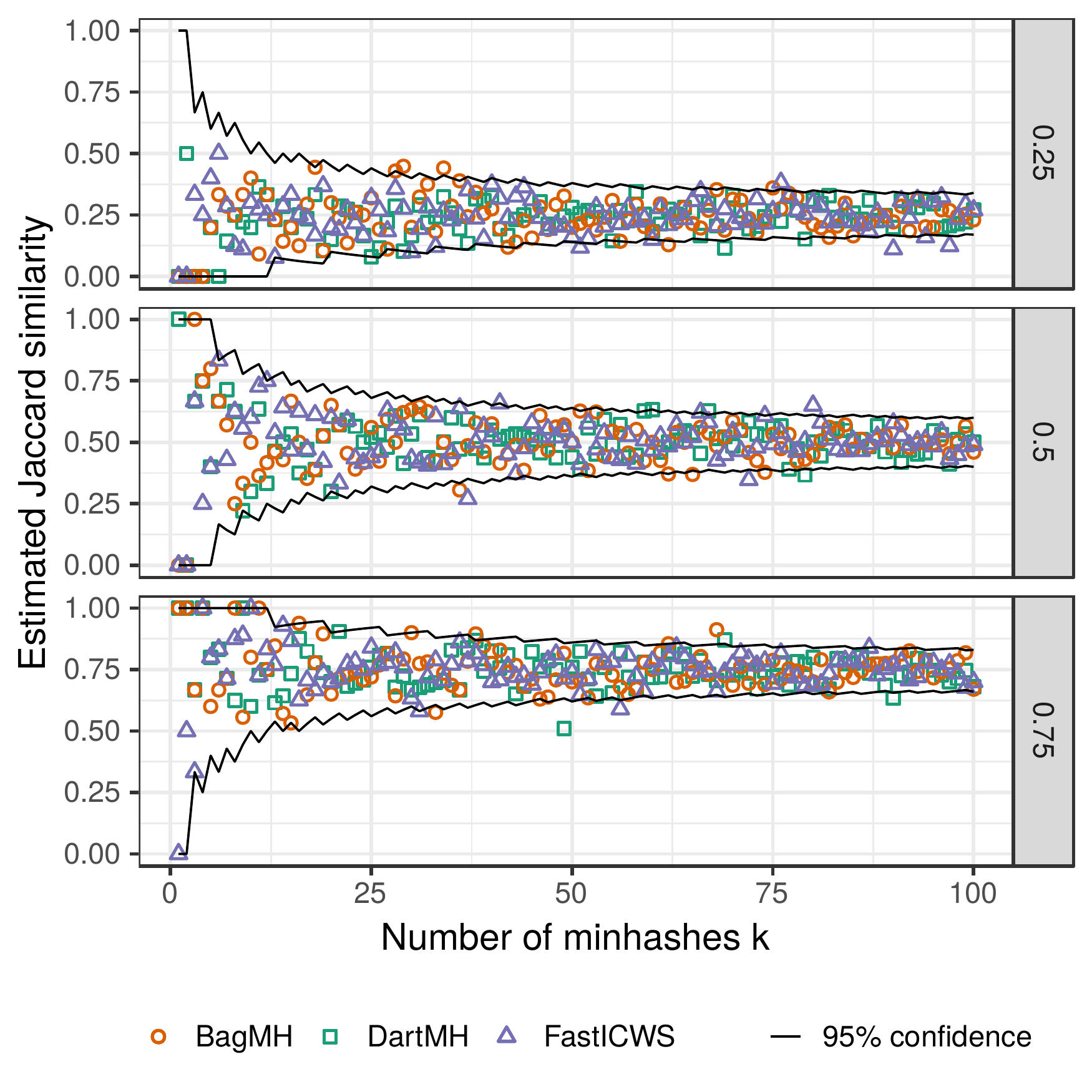}
  \caption{Estimation using $k = 1, \dots, 100$ minhashes from FastICWS, BagMinHash, and DartMinHash for Jaccard similarity $0.25$, $0.5$, and $0.75$}
  \label{fig:estimation}
\end{wrapfigure}

To benchmark the BagMinHash algorithm we use the implementation in C\texttt{++} that was made available by the author on GitHub~\cite{ertl2018bagminhashrepo}.
The BagMinHash algorithm comes in two flavors, BagMinHash1 and BagMinHash2, that further rely on a choice of discretization (float or double) of the weights.
We chose to use BagMinHash2 with the float discretization as preliminary experiments has shown this to always be the faster choice.
The BagMinHash implementation uses XXHash64~\cite{collet2016xxhash} as its source of pseudorandomness.

For ICWS we first created a straightforward and unoptimized implementation of the pseudocode in the original paper~\cite{ioffe2010improved} using tabulation hashing as the source of randomness.
The BagMinHash repository contained a more optimized implementation using the ziggurat algorithm for efficient random sampling~\cite{marsaglia2000ziggurat} that typically ran $2$x faster than our own reference implementation.
However, this implementation still proved slower than the DartMinHash algorithm for $k = 1$, so we implemented a highly optimized version of ICWS that relies on tabulation of the Gamma(2,1) distribution and minimizes the use of logarithms, exponentials, and divisions.
Furthermore, we employ the simple optimization of only computing the logarithm of each weight $x_i$ once when producing $k$ sketches, compared to $k$ times for each weight in the standard ICWS algorithm. 
The FastICWS implementation is typically $3-10$x faster than the ICWS implementation from the BagMinHash repository as can be seen from Table~\ref{tab:performance}.

There is a line of work~\cite{wu2016canonical, wu2017consistent,wu2019improved} of modifications to the ICWS algorithm to provide (relatively minor) speedups in practice. 
It was shown in ~\cite{ertl2018bagminhash} that the correctness of these alternative ICWS algorithms is questionable, and in any case they do not change the asymptotic complexity $O(k \norm{x}_0)$ of the ICWS algorithm. 
We therefore choose to restrict our focus to the ICWS implementation from the BagMinHash repository and our own FastICWS implementations of the original ICWS algorithm.
We will verify the correctness of the latter algorithm experimentally on synthetic data.

Experiments ran on an Intel i7-10510U CPU with 8mb of cache and 16gb RAM under Ubuntu 18.04 LTS and were compiled using \texttt{gcc} version 7.5.0 with optimization flags \texttt{-O3} and \texttt{-march=native}. 

\begin{table}
  \centering
  \caption{Running times in ms when varying $k$, $\norm{x}_0$, and $\norm{x}_1$}
  \label{tab:performance}
  \begin{tabular}{rrr|rrrrr}
    \toprule
    $k$  & $\norm{x}_0$ & $\norm{x}_1$ & ICWS    & FastICWS      & BagMH & DartMH   \\ \midrule
    1    & 256          & 1            & 0.04    & \textbf{0.01} & 0.10       & 0.02          \\
    256  & 256          & 1            & 7.65    & 0.92          & 3.20       & \textbf{0.17} \\
    4096 & 256          & 1            & 120.92  & 18.89         & 89.66      & \textbf{2.67} \\
    1    & 4096         & 1            & 0.63    & \textbf{0.11} & 1.53       & 0.30          \\
    256  & 4096         & 1            & 121.26  & 8.68          & 6.59       & \textbf{0.49} \\
    4096 & 4096         & 1            & 1919.49 & 152.92        & 113.30     & \textbf{3.64} \\ \midrule
    64   & 64           & 1            & 0.48    & 0.05          & 0.63       & \textbf{0.04} \\
    64   & 1024         & 1            & 7.55    & 0.50          & 1.38       & \textbf{0.10} \\
    64   & 16384        & 1            & 138.89  & 8.75          & 10.45      & \textbf{1.37} \\
    1024 & 64           & 1            & 7.79    & 2.49          & 13.33      & \textbf{0.58} \\
    1024 & 1024         & 1            & 120.41  & 12.59         & 17.19      & \textbf{0.78} \\
    1024 & 16384        & 1            & 1926.13 & 131.38        & 34.69      & \textbf{2.06} \\ \midrule
    256  & 1024         & 1            & 30.19   & 2.68          & 4.25       & \textbf{0.26} \\
    256  & 1024         & $2^{64}$     & 30.19   & 2.64          & 4.22       & \textbf{2.47} \\
    256  & 1024         & $2^{-64}$    & 30.31   & 2.68          & 4.15       & \textbf{2.25} \\
    256  & 1024         & $2^{512}$    & 30.20   & \textbf{2.75} & (3.55)     & 18.58         \\
    256  & 1024         & $2^{-512}$   & 30.19   & \textbf{2.73} & (0.01)     & 17.10         \\ \bottomrule
  \end{tabular}
\end{table}
\paragraph{Synthetic data.}
Following the approach in~\cite{ertl2018bagminhash} we will use synthetic data for our experiments since it allows us to study in detail how the running time of the algorithms depend different values of $k$, $\norm{x}_0$ and $\norm{x}_1$. 
A weighted set will be represented by a list of pairs $(i, x_i)$ where $i$ is encoded as a $64$-bit integer and $x_i$ is a double precision floating point number. 
We generate a random weighted set $x$ with desired norms $\norm{x}_0$ and $\norm{x}_1$ by sampling a point $\tilde{x} \in [0, 1]^{\norm{x}_1}$
uniformly at random from the $(\norm{x}_0) - 1)$-dimensional probabilistic simplex~\cite{smith2004sampling}, 
scaling it to have norm $\norm{x}_1$, and assigning each scaled component of $\tilde{x}$ to a unique random $64$-bit element $i$.
In other words, our synthetic data points are scaled random discrete probability distributions with uniform random support over $\{0, \dots, 2^{64} - 1\}$ of size $\norm{x}_0$.
Given a point $x$ we generate a point $y$ with desired similarity $B(x, y) = b$ and $\norm{y}_1 = \norm{x}_1$ by setting $y_i = b x_i$ and $y_{j} = 1 - b$ for a random $j$ where $x_j = 0$.
Note that according to our analysis the DartHash algorithm has worst-case performance on sets with uniform weights, so our synthetic data should provide a good indication of worst-case performance.
\subsection{Estimation error}
Figure \ref{fig:estimation} shows how the Jaccard similarity estimates behave for Jaccard similarity $0.25$, $0.5$, and $0.75$ when varying the number of minhashes.
For each value of $k = 1, \dots, 100$ we generate a fresh pair of random points $(x, y)$ with target Jaccard similarity, e.g. $J(x, y) = 0.5$, 
initialize our sketching algorithms with new random numbers, generate sketches of length $k$, and estimate the Jaccard similarity from the sketches. 
From Figure \ref{fig:estimation} we can see that the three sketching algorithms behave roughly as we would expect according to the theoretical behaviour of independent minhash values: 
The estimation error decreases as we increase the sketch length and the estimates of all three algorithms stay within the confidence intervals about 95\% of the time.
 
\subsection{Running time comparison}
We compare the performance of the different algorithms by measuring their time to create $k$ minhashes from a weighted set with a given $L_0$- and $L_1$-norm. 
For every choice of $k$, $\norm{x}_0$, and $\norm{x}_1$ we generate 100 random sets and measure average running times.

We focus most of our experiments on sketch lengths $k$ that are powers of two between $64$ and $4096$ as these are most common in retrieval and machine learning applications~\cite{li2017linearized,mitzenmacher2014efficient,satuluri2012bayesian, li2012one}. 
For the choice of the $L_0$-norm of our synthetic data we keep it roughly in the same range as $k$ as dimensionality or number of nonzeroes between $10-10000$ is common in both dense and sparse data~\cite{aumuller2017annbenchmarks,li2017linearized, mann2016}.
We choose to conduct most of our experiments on synthetic data points with $\norm{x}_1 = 1$. 
We believe that this setting most accurately reflects performance on real-world problems where data is typically scaled/normalized.

Table \ref{tab:performance} gives an overview of the running times in different settings that are meant to demonstrate the strengths and weaknesses of each algorithm.
Figure \ref{fig:performance} shows the effect on the running time of varying each of $k$, $\norm{x}_0$, $\norm{x}_1$ in turn while keeping the others fixed.
Finally, Figure \ref{fig:speedup} shows the speedup that DartMinHash is able to achieve over the state-of-the-art as we vary $k$ and $\norm{x}_0$ for $\norm{x}_1 = 1$.

Overall, we see that the running time of each algorithm behaves as we would expect from its complexity bounds. 
Looking at Table \ref{tab:performance}, the DartMinHash algorithm is on average $10$x faster than FastICWS and $15$x faster than BagMinHash.
As expected, DartMinHash performs better relative to the other algorithms when $k$ and $\norm{x}_0$ are both large and $\norm{x}_1 = 1$. 
For instance, we see a $15$x speedup for $k = \norm{x}_0 = 1024$ and a $30$x speedup for $k = \norm{x}_0 = 4096$ compared to the runner-up.
\begin{wrapfigure}[17]{r}{0.5\textwidth}
  \centering
  \includegraphics[width=0.5\textwidth]{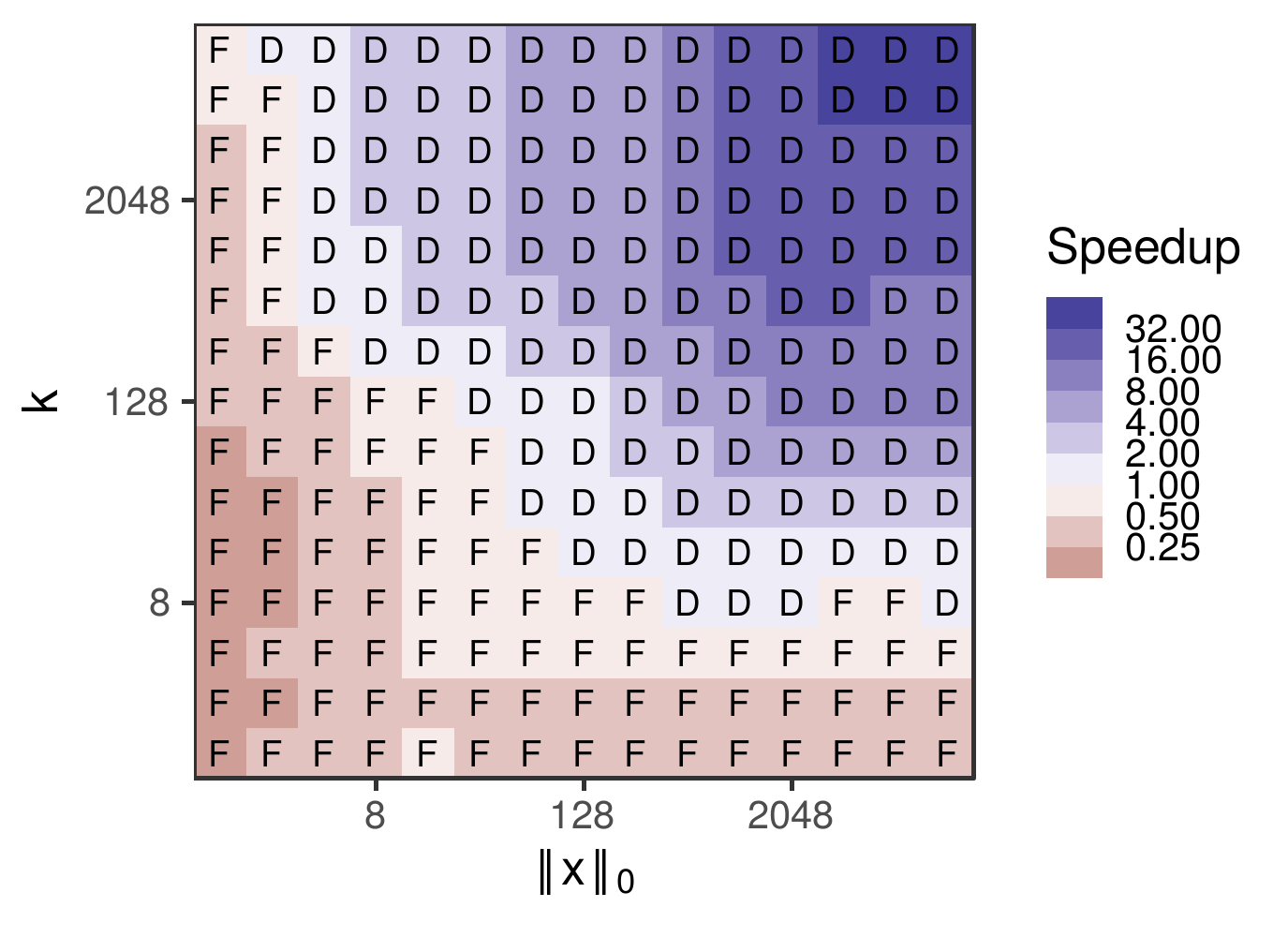}
  \caption{Speedup of DartMinHash relative to the faster of FastICWS and BagMinHash for $\norm{x}_1 = 1$}
  \label{fig:speedup}
\end{wrapfigure}

Figure \ref{fig:speedup} shows the speedup of DartMinHash (D) over FastICWS (F) and BagMinHash (B) as we vary $k$ and $\norm{x}_0$ with the fastest algorithm indicated in each setting.
We see that DartMinHash is faster than the current state of the art for all values of $k \geq 32, \norm{x}_0 \geq 64$, making it the preferred algorithm in most settings.
The performance of DartMinHash is slower than FastICWS when $k = 1$ as we would expect, but only $2-4$x slower (except when $\norm{x}_0 = 1$).
\begin{figure}
  \centering
  \subfloat[$k = 1, 2, \dots, 65536$\label{fig:variation_k}]{{\includegraphics[width=0.33\columnwidth]{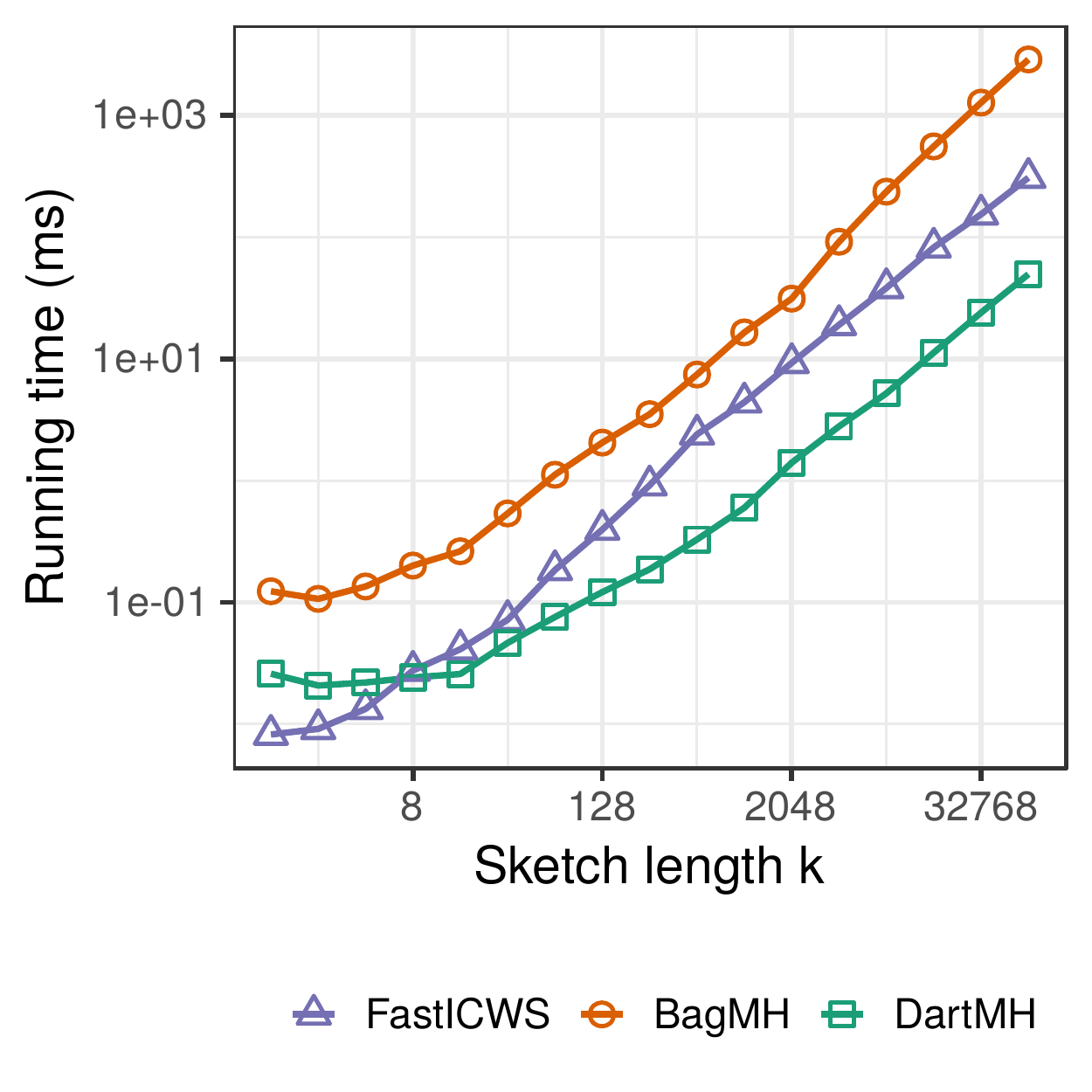}}}
  \hfill
  \subfloat[$\norm{x}_0 = 16, 32, \dots, 65536$\label{fig:variation_L0}]{{\includegraphics[width=0.33\columnwidth]{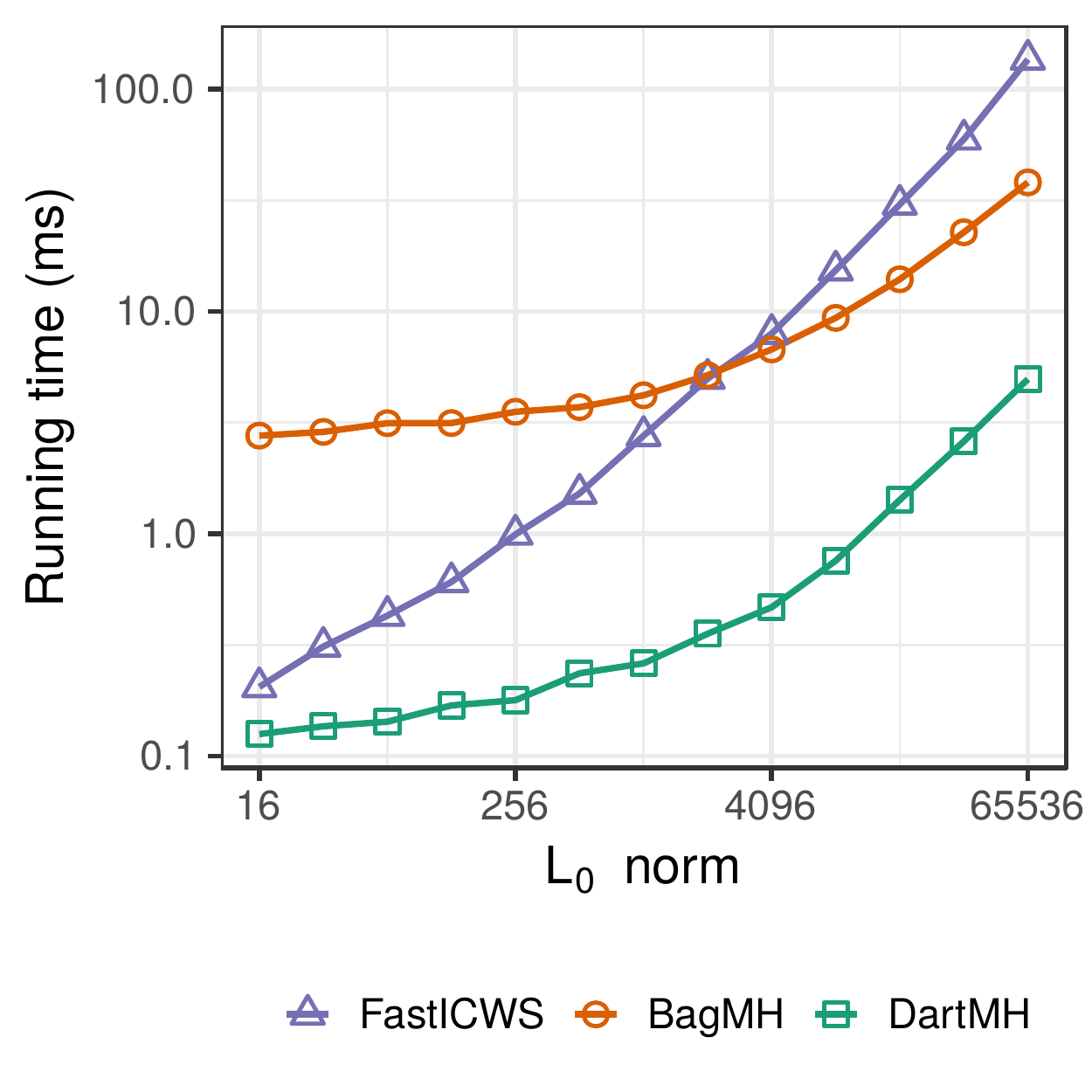}}}
  \hfill
  \subfloat[$\norm{x}_1 = 2^{-128}, 2^{-112}, \dots, 2^{128}$\label{fig:variation_L1}]{{\includegraphics[width=0.33\columnwidth]{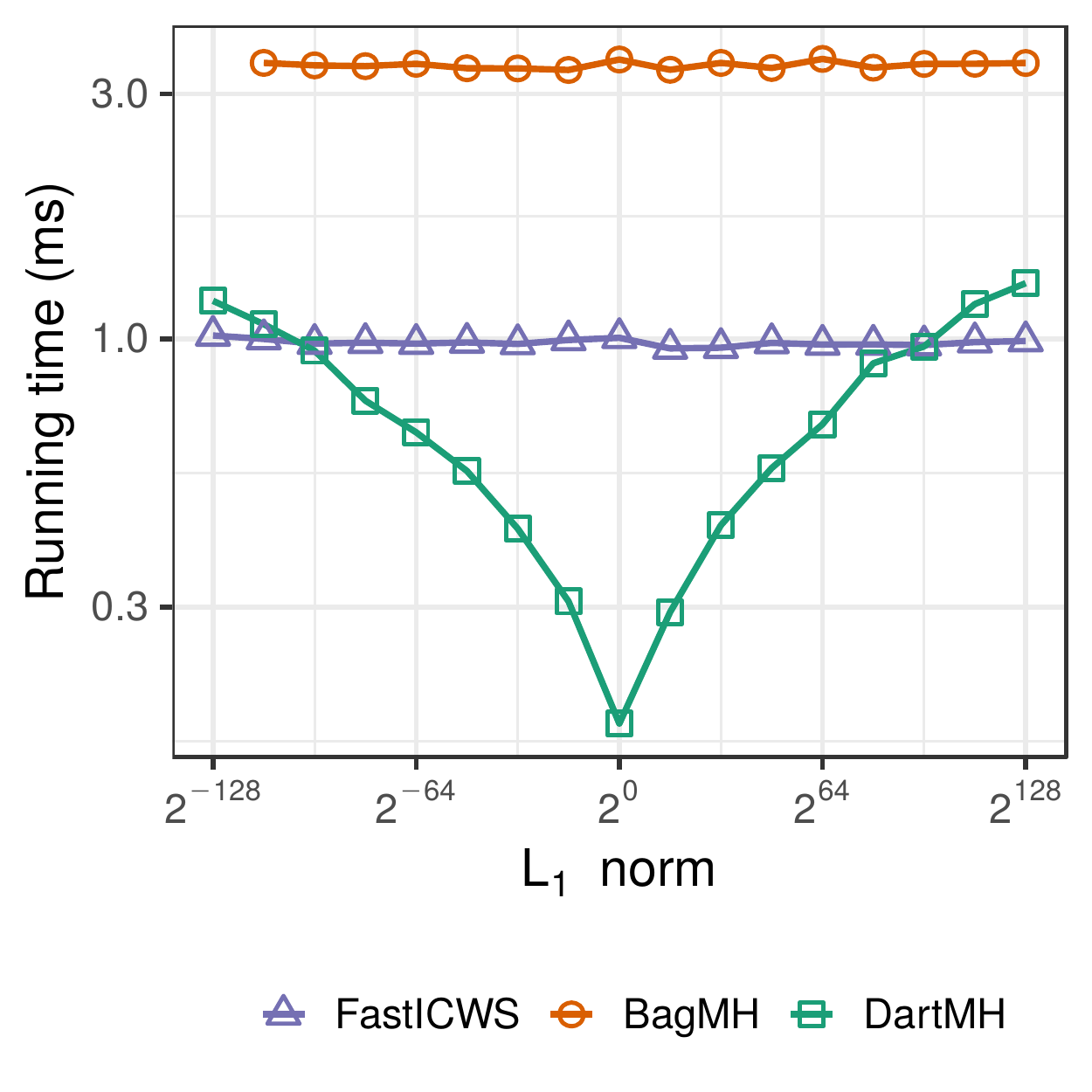}}}
  \caption{Running time when varying one of $k = 256$, $\norm{x}_0 = 256$, $\norm{x}_1 = 1$}
  \label{fig:performance}
\end{figure}

From the running time bound in Theorem~\ref{thm:minhash} we would expect the performance of DartMinHash to degrade for extreme values of $\norm{x}_1$.
From the last group of rows in Table \ref{fig:performance} we can indeed see that while DartMinHash achieves about a $10$x speedup for $k = 256, \norm{x}_0 = 1024$ in the case of $\norm{x}_1 = 1$,
when we set $\norm{x}_1 = 2^{64}, 2^{-64}$ then the performance of DartMinHash essentially matches that of FastICWS and is only about $2$x faster than BagMinHash.
At the extremes of $\norm{x}_1 = 2^{512}, 2^{-512}$ we see that FastICWS becomes about $6-7$x faster than DartMinHash. 
The running times for BagMinHash are invalid at this latter setting and have been parenthesized since we are using the floating point discretization and standard floating point does not support values outside $2^{\pm 128}$ whereas the limit for double precision is roughly $2^{\pm 1024}$. 
As expected we see that both FastICWS and BagMinHash is relatively unaffected by changes to $\norm{x}_1$ within the supported ranges.
Figure \ref{fig:variation_L1} further shows that DartMinHash retains its large speedup over FastICWS and BagMinHash for values of $\norm{x}_1$ that are not too extreme, say $2^{-32} \leq \norm{x}_1 \leq 2^{32}$.
\section*{Broader Impact}
The contribution in this paper consists of an algorithm that can be applied to speed up sketching and similarity estimation of weighted sets.
Algorithms for sketching sets have been widely known and been in use for over 20 years, so I mostly see this work as a potential efficiency improvement.
Positive impacts could be to reduce power consumption and storage requirements in systems performing similarity search or learning using large-scale kernel machines.
If similarity search is a fundamental subtask of artificial general intelligence, which I find likely, 
then this research can potentially be used to bring the day of an intelligence explosion closer, but the same can be said for many algorithmic improvements.  
\begin{ack}
I would like to thank Rasmus Pagh for his comments on an earlier version of the paper.

The research leading to these results has received funding from the European Research Council under the European Union’s 7th Framework Programme (FP7/2007-2013) / ERC grant agreement no. 614331. 
This research also received direct funding from the Norwegian Open AI Lab at the Norwegian University of Science and Technology. 
\end{ack}
\bibliographystyle{abbrv}
\bibliography{sketch}
\end{document}